\documentclass{llncs}
\usepackage{graphics,graphicx}
\newtheorem{redrule}{Rule}
\newtheorem{observation}{Observation} 

\begin{document}

\title{Contracting graphs to paths and trees\thanks{Supported by the Research Council of Norway (project SCOPE, 197548/V30) and the French ANR (project AGAPE, ANR-09-BLAN-0159).}}

\author{
Pinar Heggernes\inst{1}
\and Pim van 't Hof\inst{1}
\and Benjamin L\'ev\^eque\inst{2}
\and Daniel Lokshtanov\inst{1}
\and Christophe Paul\inst{2}}

\institute{Department of Informatics, University of Bergen, N-5020 Bergen, Norway.\\
\email{\{pinar.heggernes,pim.vanthof,daniel.lokshtanov\}@ii.uib.no}
\and CNRS, LIRMM, Universit{\'e} Montpellier 2, Montpellier, France\\
\email{\{leveque,paul\}@lirmm.fr}
}

\maketitle
\pagestyle{plain}

\begin{abstract}
Vertex deletion and edge deletion problems play a central role in Parameterized Complexity. Examples include classical problems like {\sc Feedback Vertex Set}, {\sc Odd Cycle Transversal}, and {\sc Chordal Deletion}. Interestingly, the study of edge contraction problems of this type from a parameterized perspective has so far been left largely unexplored. We consider two basic edge contraction problems, which we call {\sc Path}-{\sc Contractibility} and {\sc Tree}-{\sc Contractibility}. Both problems take an undirected graph $G$ and an integer $k$ as input, and the task is to determine whether we can obtain a path or an acyclic graph, respectively, by contracting at most $k$ edges of $G$. Our main contribution is an algorithm with running time $4^{k+O(\log^2 k)} + n^{O(1)}$ for {\sc Path}-{\sc Contractibility} and an algorithm with running time $4.88^k \, n^{O(1)}$ for {\sc Tree}-{\sc Contractibility}, based on a novel application of the color coding technique of Alon, Yuster and Zwick. Furthermore, we show that {\sc Path}-{\sc Contractibility} has a kernel with at most $5k+3$ vertices, while {\sc Tree}-{\sc Contractibility} does not have a polynomial kernel unless coNP $\subseteq$ NP/poly. We find the latter result surprising, because of the strong connection between {\sc Tree}-{\sc Contractibility} and {\sc Feedback Vertex Set}, which is known to have a vertex kernel with size $O(k^2)$.
\end{abstract}

\section{Introduction}

The $\Pi$-{\sc Contractibility} problem takes as input a graph $G$ and an integer $k$, and the question is whether there is a graph $H$ belonging to the graph class $\Pi$ such that $G$ can be contracted to $H$ using at most $k$ edge contractions. In early papers by Watanabe et al.~\cite{WAN81,WAN83} and Asano and Hirata~\cite{AH83}, $\Pi$-{\sc Contractibility} was proved to be NP-complete for several classes $\Pi$. The $\Pi$-{\sc Contractibility} problem fits into a wider and well studied family of graph modification problems, where vertex deletions and edge deletions are two other ways of modifying a graph. $\Pi$-{\sc Vertex Deletion} and $\Pi$-{\sc Edge Deletion} are the problems of deciding whether some graph belonging to graph class $\Pi$ can be obtained from $G$ by $k$ vertex deletions or by $k$ edge deletions, respectively. All of these problems are shown to be NP-complete for most of the interesting graph classes $\Pi$ \cite{NSS01,Yannakakis78,Yannakakis79,Yannakakis81}. However, whereas $\Pi$-{\sc Vertex Deletion} and $\Pi$-{\sc Edge Deletion} have been studied in detail for several graph classes $\Pi$ with respect to fixed parameter tractability~\cite{B94,Cai96,DF99,Marx06,MS07,PRY10}, no such result has been known for $\Pi$-{\sc Contractibility} until now. Very recently, four of the authors proved that $\Pi$-{\sc Contractibility} is fixed parameter tractable when $\Pi$ is the class of bipartite graphs~\cite{HHLP11}, which is, to our knowledge, the only result of this type so far.

Here we study $\Pi$-{\sc Contractibility} when $\Pi$ is the class of acyclic graphs and when $\Pi$ is the class of paths.
Observe that edge contractions preserve the number of connected components. Hence, a graph can be contracted to an acyclic graph 
using $k$ edge contractions if and only if each of its connected components can be contracted to a tree, 
using at most $k$ edge contractions in total. Consequently, for the problems that we consider in this paper, we may assume that the input graph is connected, and it is appropriate to name our problems {\sc Tree-Contractibility} and {\sc Path-Contractibility}.
Both problems are NP-complete, which easily follows from previous results~\cite{AH83,BV87}. We find these problems of particular interest, since their vertex deletion and edge deletion versions are famous and well studied; these problems are widely known as {\sc Feedback Vertex Set}, {\sc Longest Induced Path}, {\sc Spanning Tree}, and {\sc Longest Path}, respectively. All of these problems, except {\sc Spanning Tree}, are known to be NP-complete~\cite{GJ79}. Furthermore, when parameterized by the number of deleted vertices or edges, they are fixed parameter tractable and have polynomial kernels.

A parameterized problem with input size $n$ and parameter $k \le n$ is {\it fixed parameter tractable (FPT)} \cite{DF99} if there is an algorithm with running time $f(k) \, n^{O(1)}$ that solves the problem, where $f$ is a function that does not depend on $n$. Such a problem has a {\it kernel} if, for every instance $(I,k)$, an instance $(I',k')$ can be constructed in polynomial time such that $|I'| \le g(k)$, and $(I,k)$ is a {\sc yes}-instance if and only if $(I',k')$ is a {\sc yes}-instance. A parameterized problem is FPT if and only if it has a kernel~\cite{DF99}. Some FPT problems have the desirable property that the size $g(k)$ of their kernel is a polynomial function. Whether an FPT problem has a polynomial kernel or not has attracted considerable attention over the last years, especially after the establishment of methods for proving non-existence of polynomial kernels, up to some complexity theoretical assumptions~\cite{BFLPST09,BDFH08,BTY09}. In fact, a problem that has been shown to be FPT enters two races: the race for a better kernel and the race for a better running time. During the last decade, considerable effort has been devoted to improving the parameter dependence in the running time of classical parameterized problems. Even in the case of a running time which is single exponential in $k$, lowering the base of the exponential function is considered to be an important challenge. For instance, the running time of {\sc Feedback Vertex Set} has been successively improved from $37.7^k \, n^{O(1)}$~\cite{GGH06} to $10.57^k \, n^{O(1)}$~\cite{DFL07}, and most recently to $5^k \, n^{O(1)}$~\cite{CFL08}.

In this paper, we present results along these established tracks for our two problems {\sc Tree-Contractibility} and {\sc Path-Contractibility}. It can be shown that if a graph $G$ is contractible to a path or a tree with at most $k$ edge contractions, then the longest cycle in $G$ is of size $O(k)$, and hence the treewidth of $G$ is bounded by a function of $k$. Consequently, when parameterized by $k$, fixed parameter tractability of {\sc Tree-Contractibility} and {\sc Path-Contractibility} follows from the well known result of Courcelle \cite{Courcelle90}, as both problems are expressible in monadic second order logic. However, this approach yields very unpractical algorithms whose running times involve huge functions of $k$. Here, we give algorithms with running time $c^k \, n^{O(1)}$ with small constants $c <5$ for both problems. To obtain our results, we use a variant of the color coding technique of Alon, Yuster and Zwick \cite{AYZ95}. Furthermore, we show that {\sc Path-Contractibility} has a linear vertex kernel. On the negative side, we show that {\sc Tree-Contractibility} does not have a polynomial kernel, unless coNP $\subseteq$ NP/poly. This is a contrast compared to the corresponding vertex deletion version, as {\sc Feedback Vertex Set} has a quadratic kernel~\cite{Thomasse09}.

\section{Definitions and notation}

All graphs in this paper are finite, undirected, and simple, i.e., do not contain multiple edges or loops. Given a graph $G$, we denote its vertex set by $V(G)$ and its edge set by $E(G)$. We also use the ordered pair $(V(G),E(G))$ to represent $G$. We let $n= |V(G)|$. Let $G=(V,E)$ be a graph. The {\em neighborhood} of a vertex $v$ in $G$ is the set $N_G(v)=\{w\in V \mid vw\in E\}$ of {\it neighbors} of $v$ in $G$. Let $S\subseteq V$. We write $N_G(S)$ to denote $\bigcup_{v\in S}N_G(v) \setminus S$. We say that $S$ {\em dominates} a set $T\subseteq V$ if every vertex in $T$ either belongs to $S$ or has at least one neighbor in $S$. We write $G[S]$ to denote the subgraph of $G$ {\em induced} by $S$. We use shorthand notation $G{-}v$ to denote $G[V \setminus \{v\}]$ for a vertex $v \in V$, and $G{-}S$ to denote $G[V \setminus S]$ for a set of vertices $S \subseteq V$. A graph is {\it connected} if it has a path between every pair of its vertices, and is {\em disconnected} otherwise. The {\it connected components} of a graph are its maximal connected subgraphs. We say that a vertex subset $S$ of a graph $G$ is {\it connected} if $G[S]$ is connected. A {\it bridge} in a connected graph is an edge whose deletion results in a disconnected graph. A {\it cut vertex} in a connected graph is a vertex whose deletion results in a disconnected graph. A graph is 2-{\em connected} if it has no cut vertex. A 2-{\em connected component} of a graph $G$ is a maximal 2-connected subgraph of $G$. We use $P_\ell$ to denote the graph isomorphic to a path on $\ell$ vertices, i.e., the graph with ordered vertex set $\{p_1, p_2,p_3, \ldots, p_\ell \}$ and edge set $\{p_1p_2, p_2p_3, \ldots,p_{\ell -1}p_\ell \}$. We will also write $p_1p_2 \cdots p_\ell$ to denote $P_\ell$. A {\it tree} is a connected acyclic graph. A vertex with exactly one neighbor in a tree is called a {\it leaf}. A {\it star} is a tree isomorphic to the graph with vertex set $\{a, v_1, v_2, \ldots, v_s\}$ and edge set $\{av_1, av_2, \ldots, av_s\}$. Vertex $a$ is called the {\it center} of the star.

The \emph{contraction} of edge $xy$ in $G$ removes vertices $x$ and $y$ from $G$, and replaces them by a new vertex, which is made adjacent to precisely those vertices that were adjacent to at least one of the vertices $x$ and $y$. 
A graph $G$ is \emph{contractible} to a graph $H$, or $H$-{\it contractible}, if $H$ can be obtained from $G$ by a sequence of edge contractions. Equivalently, $G$ is $H$-contractible if there is a surjection $\varphi: V(G) \rightarrow V(H)$, with $W(h) = \{v \in V(G) \mid \varphi(v) = h\}$ for every $h \in V(H)$, that satisfies the following three conditions: {\bf (1)} for every $h \in V(H)$, $W(h)$ is a connected set in $G$; {\bf (2)} for every pair $h_i,h_j\in V(H)$, there is an edge in $G$ between a vertex of $W(h_i)$ and a vertex of $W(h_j)$ if and only if $h_ih_j \in E(H)$; {\bf (3)} ${\cal W}=\{W(h) \mid h \in V(H)\}$ is a partition of $V(G)$. We say that ${\cal W}$ is an $H$-{\it witness structure} of $G$, and the sets $W(h)$, for $h\in V(H)$, are called {\it witness sets} of ${\cal W}$. It is easy to see that if we contract every edge $uv \in E(G)$, such that $u$ and $v$ belong to the same witness set, then we obtain a graph isomorphic to $H$. Hence $G$ is $H$-contractible if and only if it has an $H$-witness structure. 

If a witness set of contains more than one vertex of $G$, then we call it a {\it big} witness set; a witness set consisting of a single vertex of $G$ is called {\em small}. We say that $G$ is \emph{k-contractible} to $H$, with $k\leq n-1$, if $H$ can be obtained from $G$ by at most $k$ edge contractions. The next observation follows from the above.

\begin{observation}
\label{obs:big}
If a graph $G$ is $k$-contractible to a graph $H$, 
then $|V(G)|-k\leq |V(H)|$, and any $H$-witness structure ${\cal W}$ of $G$ satisfies the following three properties: no witness set of ${\cal W}$ contains more than $k+1$ vertices, ${\cal W}$ has at most $k$ big witness sets, and the big witness sets of ${\cal W}$ altogether contain at most $2k$ vertices.
\end{observation}

A $2$-{\em coloring} of a graph $G$ is a function $\phi : V(G) \rightarrow \{1,2\}$. We point out that a 2-coloring of $G$ is merely an assignment of colors 1 and 2 to the vertices of $G$, and should therefore not be confused with a {\em proper} 2-coloring of $G$, which is a 2-coloring with the additional property that no two adjacent vertices receive the same color. If all the vertices belonging to a set $S \subseteq V(G)$ have been assigned the same color 
by $\phi$, we say that $S$ is {\it monochromatic} with respect to $\phi$, and we use $\phi(S)$ to denote the color of the
vertices of $S$. Any 2-coloring $\phi$ of $G$ defines a partition of $V(G)$ into two sets $V_\phi^1$ and $V_\phi^2$, which are the sets of vertices of $G$ colored 1 and 2 by $\phi$, respectively. A set $X\subseteq V(G)$ is a {\em monochromatic component} of $G$ with respect to $\phi$ if $G[X]$ is a connected component of $G[V_\phi^1]$ or a connected component of $G[V_\phi^2]$. We say that two different 2-colorings $\phi_1$ and $\phi_2$ of $G$ {\it coincide} on a vertex set $A \subseteq V(G)$ if $\phi_1(v) = \phi_2(v)$ for every vertex $v \in A$.

\section{\sc Path-Contractibility}
\label{sec:path}

Brouwer and Veldman~\cite{BV87} showed that it is NP-complete to decide whether a graph can be contracted to the path $P_\ell$, for every fixed $\ell\geq 4$. This, together with the observation that a graph $G$ is $k$-contractible to a path if and only if $G$ is contractible to $P_{n-k}$, implies that {\sc Path-Contractibility} is NP-complete. In this section, we first show that {\sc Path-Contractibility} has a linear vertex kernel. We then present an algorithm with running time $4^{k+O(\log^2 k)} + n^{O(1)}$ for the same problem. Throughout this section, whenever we mention a $P_\ell$-witness structure ${\cal W} = \{W_1, \ldots W_\ell\}$, it will be implicit that $P_\ell=p_1 \cdots p_\ell$, and $W_i=W(p_i)$ for every $i\in \{1,\dots, \ell\}$. We start by formulating a reduction rule, and then show that it is ``safe''.

\begin{redrule}
\label{rule:1}
Let $(G,k)$ be an instance of {\sc Path-Contractibility}. If $G$ contains a bridge $uv$ such that the deletion of edge $uv$ from $G$ results in two connected components that contain at least $k+2$ vertices each, then return $(G',k)$, where $G'$ is the graph resulting from the contraction of edge $uv$.
\end{redrule}

\begin{lemma}
\label{lem:rulesafe}
Let $(G',k)$ be an instance of {\sc Path-Contractibility} resulting from the application of Rule~\ref{rule:1} on $(G,k)$. Then $G'$ is $k$-contractible to a path if and only if $G$ is $k$-contractible to a path.
\end{lemma}

\begin{proof}
Let $G$ be a graph on which Rule 1 is applicable, and let $uv$ be the bridge of $G$ that is contracted to obtain $G'$.  Let $G_1$ and  $G_2$ be the two connected components that we obtain if we delete the edge $uv$ from $G$, with $L=V(G_1)$ and $R=V(G_2)$, such that $u \in L$ and $v \in R$. Furthermore, let $L' =L \setminus \{u\}$ and $R'= R \setminus \{v\}$, and let $w$ be the vertex of $G'$ resulting from the contraction of $uv$ in $G$.

Assume that $G$ is $k$-contractible to a path $P_\ell$, and let ${\cal W}=\{W_1, \ldots, W_\ell\}$ be a $P_\ell$-witness structure of $G$.  If $u$ and $v$ belong to the same witness set $W_i$ of ${\cal W}$, then we can obtain a $P_\ell$-witness structure ${\cal W'}$ of $G'$ by replacing $W_i$ with a new set $W_i'=(W_i\setminus\{u,v\})\cup\{w\}$, and keeping all other witness sets of ${\cal W}$ the same. Hence $G'$ can be contracted to $P_\ell$. Since $G$ is $k$-contractible to $P_\ell$, and $G'$ has one fewer vertex than $G$, we know that $G'$ is $(k-1)$-contractible to a path. If $u$ and $v$ belong to two different witness sets of ${\cal W}$, then one belongs to $W_i$ and the other to $W_{i+1}$, for two adjacent vertices $p_i$ and $p_{i+1}$ of $P_\ell$. Furthermore, $uv$ is the only edge in $G$ between a vertex of $A=\bigcup_{j=1}^{i} W_j$ and a vertex of $B=\bigcup_{j=i+1}^{\ell} W_j$, since both $G[A]$ and $G[B]$ are connected and $uv$ is a bridge of $G$. Consequently, by replacing $W_i$ and $W_{i+1}$ by one witness set $W_i''=((W_i \cup W_{i+1}) \setminus \{u,v\}) \cup \{w\}$, and keeping all other witness sets of ${\cal W}$ the same, we obtain a $P_{\ell -1}$-witness structure ${\cal W''}$ of $G'$.  Hence $G'$ is $k$-contractible to a path.

For the other direction, assume that $G'$ is $k$-contractible to a path $P_{\ell'}$, and let ${\cal W}=\{W_1, \ldots, W_{\ell'}\}$ be a $P_{\ell'}$-witness structure of $G'$. Let $W_i$ be the witness set of ${\cal W}$ containing $w$. The set $A = \bigcup_{j=1}^{i-1} W_j$ is connected in $G'{-}w = G{-}\{u,v\}$. Hence $A$ cannot contain vertices from both $L'$ and $R'$, so it contains only elements of $L'$ or only elements of $R'$.  Similarly, the set $B = \bigcup_{j=i+1}^{\ell} W_j$ contains only elements of $L'$, or only elements of $R'$. By Observation~\ref{obs:big}, the set $W_i$, which contains $w$, does not contain all of $L'$, as $|L'\cup\{w\}|\geq k+2$. Similarly, $W_i$ does not contain all of $R'$, as $|R'\cup\{w\}|\geq k+2$. Hence, neither $A$ nor $B$ is empty, one contains only elements of $R'$ and the other only elements of $L'$.  Consequently, by replacing $W_i$ by two new sets $(W_i\cap L')\cup \{u\}$ and $(W_i\cap R')\cup \{v\}$, and keeping all other witness sets of ${\cal W}$ the same, we obtain a $P_{\ell'+1}$-witness structure of $G$. Hence $G$ is $k$-contractible to a path.
\qed
\end{proof}

\begin{theorem}
\label{thm:path-linear-kernel}
{\sc Path-Contractibility} has a kernel with at most $5k + 3$ vertices.
\end{theorem}

\begin{proof}
Let $(G,k)$ be an instance of {\sc Path-Contractibility}. Starting from $(G,k)$, we repeatedly test, in linear time, whether Rule~\ref{rule:1} can be applied on the instance under consideration, and apply the reduction rule if possible. Each application of Rule 1 strictly decreases the number of vertices. Hence, starting from $G$, we reach a \emph{reduced} graph on which Rule~\ref{rule:1} cannot be applied anymore in polynomial time. By Lemma~\ref{lem:rulesafe}, we know that the resulting reduced graph is $k$-contractible to a path if and only if $G$ is $k$-contractible to a path.

We now assume that $G$ is reduced. We show that if $G$ is $k$-contractible to a path, then $G$ has at most $5k+3$ vertices. Let $\mathcal{W} = \{W_1, \ldots, W_\ell\}$ be a $P_\ell$-witness structure of $G$ with $\ell \geq n-k$.  We first prove that $\ell \leq 4k+3$. Assume that $\ell \ge 2k+4$, and let $i$ be such that $k+2 \leq i \leq \ell-k-2$. Suppose, for contradiction, that both $W_i$ and $W_{i+1}$ are small witness sets, i.e., $W_i=\{u\}$ and $W_{i+1}=\{v\}$ for two vertices $u$ and $v$ of $G$. Then $uv$ forms a bridge in $G$ whose deletion results in two connected components. Each of these components contains at least all vertices from $W_1,\ldots,W_{k+2}$ or all vertices from $W_{\ell-k-1},\ldots, W_{\ell}$. Hence they contain at least $k+2$ vertices each. Consequently, Rule~\ref{rule:1} can be applied, contradicting the assumption that $G$ is reduced. So there are no consecutive small sets among $W_{k+2},\ldots,W_{\ell-k-1}$. By Observation~\ref{obs:big}, $\mathcal{W}$ contains at most $k$ big witness sets, so we have $(\ell-k-1)-(k+2)+1\leq 2k+1$ implying $\ell\leq 4k+3$. Since $n-k \leq \ell$ and $\ell \le 4k+3$, we conclude that $n \le 5k+3$.
\qed
\end{proof}

The existence of a kernel with at most $5k+3$ vertices immediately implies the following algorithm with running time $32^{k} \, n^{O(1)}$ for {\sc Path-Contractibility}. Generate all the different 2-colorings of the reduced graph $G$; there are at most $2^{5k+3}$ of those colorings. For each of these, check whether defining every monochromatic component as a witness set results in a $P_\ell$-witness structure of $G$, for $\ell \ge n-k$. This check can clearly be performed in polynomial time. It is easy to verify that the described algorithm will find a $P_\ell$-witness structure of $G$ if and only if $G$ is contractible to $P_\ell$, and the running time follows. The natural follow-up question is how much we can improve this running time. In this section we give an algorithm with running time $4^{k + O(\log^2 k)} + n^{O(1)}$. To that aim, we will first present a Monte Carlo algorithm with running time $4^{k + O(\log k)} + n^{O(1)}$, and then derandomize this algorithm at the cost of a slightly worse running time and exponential space.

Let $G$ be a graph, and let ${\cal W}=\{W_1, \ldots, W_\ell\}$ be a $P_\ell$-witness structure of $G$. We say that a $2$-coloring $\phi$ of $G$ is {\em compatible} with ${\cal W}$ (or $\mathcal{W}$-\emph{compatible}) if the following two conditions are both satisfied: {\bf (1)} every witness set of ${\cal W}$ is monochromatic with respect to $\phi$, and {\bf (2)} if $W_i$ and $W_j$ are big witness sets with $i < j$, and all witness sets $W_{i'}$ with $i < i' < j$ are small, then $\phi(W_i) \neq \phi(W_j)$. The randomized algorithm for {\sc Path-Contractibility}, presented in the proof of Theorem~\ref{thm:path-random} below, starts by generating a number of random 2-colorings of $G$. We show that, if $G$ is $k$-contractible to a path $P_\ell$, then at least one of these random 2-colorings is compatible with a $P_\ell$-witness structure of $G$, with relatively high probability. The next lemma shows that, if a 2-coloring $\phi$ is compatible with a $P_\ell$-witness structure of $G$, then we can compute a $P_{\ell'}$-witness structure of $G$, such that $\ell'\geq \ell$, in polynomial time. 

\begin{lemma}\label{lem:pathalg1} 
Let $\phi$ be a $2$-coloring of a graph $G$. If $\phi$ is compatible with a $P_\ell$-witness structure of $G$, then a $P_{\ell'}$-witness structure of $G$ can be computed in polynomial time, such that $\ell' \ge \ell$.
\end{lemma}

\begin{proof}
Suppose that $\phi$ is compatible with a $P_\ell$-witness structure ${\cal W}$ of $G$, with ${\cal W}=\{W_1, \ldots, W_\ell\}$. Let ${\cal X}=\{X_1,\ldots,X_r\}$ be the set of monochromatic components of $G$ with respect to $\phi$. Clearly, $r \le \ell$. By property (1) of a ${\cal W}$-compatible coloring, the big witness sets of $\mathcal W$ are monochromatic with respect to $\phi$. They are also connected by definition. Hence, for every big witness set $W_j\in {\mathcal W}$, there is an $i$ with $1\leq i \leq r$, such that $W_j\subseteq X_i$. Moreover, if a set $X_i$ contains several witness sets of ${\cal W}$, then these witness sets $W_{i_1}, \ldots, W_{i_q}$ are consecutive witness sets of ${\cal W}$, i.e., $p_{i_1} \cdots p_{i_q}$ is a subpath of $P_\ell$. Consequently, $\mathcal X$ is a $P_r$-witness structure of $G$. We can assume without loss of generality that the indexing of the witness sets of ${\cal X}$ respects the indexing of the witness sets of ${\cal W}$, i.e., if $W_j \subseteq  X_i$ and $W_{j'} \subseteq X_{i'}$ and $i<i'$, then $j<j'$. We now show how to decompose some of the witness sets of ${\cal X}$ into smaller witness sets, if $r < \ell$.

Let $i\in \{1,\ldots, r\}$. Let $1\leq a\leq b\leq \ell$, such that $X_i=\bigcup_{j=a}^{b} W_j$. By property (2) of a ${\cal W}$-compatible coloring, $X_i$ contains at most one big witness set of $\mathcal{W}$. In other words, $G[X_i]$ has a $P_{s}$-witness structure, with $s = b-a+1$, containing at most one big witness set. Define $X_0=X_{r+1}=\emptyset$, and let $L=N_G(X_{i-1})\cap X_i$ and $R=N_G(X_{i+1})\cap X_i$. Thus we have $L\subseteq W_{a}$ and $R\subseteq W_{b}$. A {\em shatter} of $X_i$ is a $P_{t}$-witness structure $\mathcal Y=\{Y_1,\ldots,Y_{t}\}$ of $G[X_i]$, for some $t \ge 1$, such that $\mathcal Y$ contains at most one big witness set, $L\subseteq Y_1$, and $R\subseteq Y_t$. The {\it size} of a shatter is the size of its biggest witness set. 
Note that $\{X_i\}$ and $\mathcal \{W_{a},\ldots,W_{b}\}$ are shatters of $X_i$, and that the size of the latter is smaller if $a\neq b$. In fact, any shatter ${\cal Y}$ of $X_i$ is of one of the following five types: (i) all witness sets of ${\cal Y}$ are small; (ii) ${\cal Y}$ contains a single witness set and this set is big; (iii) the only big witness set of ${\cal Y}$ has neighbors in $X_{i-1}$ but not in $X_{i+1}$; (iv) the only big witness set of ${\cal Y}$ has neighbors in $X_{i+1}$ but not in $X_{i-1}$; (v) the only big witness set of ${\cal Y}$ does not have neighbors in $X_{i+1}$ or $X_{i-1}$.

Recall that, given $\phi$, we do not know ${\cal W}$, but only ${\cal X}$. Let $X_i$ be a big witness set of ${\cal X}$. We will now explain how to find a minimum size shatter of $X_i$ in polynomial time. The idea is to check, for each of the five possible types of shatters listed above, if $X_i$ has a shatter of that type, and to compute a shatter of minimum size of that type if one exists. Comparing the sizes of each of those shatters clearly yields a minimum size shatter of $X_i$. For convenience, we assume that $i\in \{2,\ldots,r-1\}$, which means that at least one vertex of $X_i$ has a neighbor in $X_{i-1}$, and at least one vertex in $X_i$ has a neighbor in $X_{i+1}$. The arguments below can easily be adapted to work for the cases $i=1$ and $i=r$. Having a shatter of type (i) is equivalent to $G[X_i]$ being a path, such that the first vertex of the path is the only vertex of $X_i$ with a neighbor in $X_{i-1}$, and the last vertex of the path is the only vertex with a neighbor in $X_{i+1}$. If this case applies, which we can check in $O(|X_i|)$ time, then the unique shatter of $X_i$ of type (i) is the shatter all whose witness sets have size 1. Note that $X_i$ always has a shatter of type (ii), and that $\{X_i\}$ is the unique shatter of that type. Having a shatter of type (v) is equivalent to the existence of a set $B\subseteq X_i$, such that the following holds: $G[X_i \setminus B]$ has exactly two connected components $P_1$ and $P_2$ that are both paths, only the first vertex of $P_1$ has a neighbor in $X_{i-1}$, only the last vertex of $P_2$ has a neighbor in $X_{i+1}$, and only the last vertex $x$ of $P_1$ and the first vertex $y$ of $P_2$ have neighbors in $B$. Observe that in this case $G[X_i \setminus \{x,y\}]$ has exactly 3 connected components, and that the size of $B$ defines the size of a shatter of type (v). We can find out whether $X_i$ has a shatter of type (v), and if so, compute a minimum size shatter of that type, by trying all pairs of vertices in $X_i$ as $x$ and $y$. This can clearly be done in $O(|X_i|^3)$ time. We can find minimum size shatters of types (iii) and (iv), or conclude that they do not exist, in a similar way, and for them it is enough to examine the connected components of $G[X_i \setminus \{v\}]$ for every vertex $v \in X_i$. 

The above procedure can now be repeated for every big witness set $X_i\in {\cal X}$, and in total time $O(n^3)$ we can find a $P_{\ell'}$-witness structure of $G$ by combining the computed shatters. Since $G$ is $P_\ell$-contractible, $\phi$ is compatible with a $P_\ell$-witness structure of $G$, and we have computed a shatter of minimum size for each big witness set of ${\cal X}$,
we have that $\ell' \ge \ell$.
\qed
\end{proof}

\begin{theorem}
\label{thm:path-random}
There is a Monte Carlo algorithm with running time $4^{k + O(\log k)} + n^{O(1)}$ for {\sc Path-Contractibility}. If the input graph $G$ is not $k$-contractible to a path, then the algorithm correctly outputs {\sc no}. If $G$ is $k$-contractible to a path, then the algorithm correctly outputs {\sc yes} with probability at least $1-\frac{1}{e}$. Moreover, if it outputs {\sc yes}, it also outputs a $P_\ell$-witness structure of $G$ with $\ell \geq n-k$.
\end{theorem}

\begin{proof}
We describe such a randomized algorithm. Let $G$ be an input graph. The algorithm has a main loop, which is repeated $4^k$ times. At each iteration of the loop, the algorithm generates a random $2$-coloring $\phi$ of $G$. The algorithm then tries to compute a $P_\ell$-witness structure of $G$ with $\ell\geq n-k$, which it will find, by Lemma~\ref{lem:pathalg1}, in polynomial time if $\phi$ is compatible with a $P_{\ell'}$-witness structure of $G$ for $\ell'\geq n-k$. If none of the $4^k$ iterations of the main loop yields a $P_\ell$-witness structure of $G$ with $\ell\geq n-k$, then we return {\sc no}. The running time of this algorithm is $ 4^k \, n^{O(1)}$. If, prior to the main loop, we compute a kernel of $G$ with at most $5k+3$ vertices in polynomial time according to Theorem~\ref{thm:path-linear-kernel}, then the running time is improved to $4^k \, k^{O(1)} + n^{O(1)} = 4^{k + O(\log k)} + n^{O(1)}$.

Suppose $G$ is $k$-contractible to a path $P_\ell$, and let ${\cal W}= \{W_1, \ldots, W_\ell\}$ be a $P_\ell$-witness structure of $G$. Let $W_{i_1}, \ldots, W_{i_q}$ be the big witness sets of ${\cal W}$, such that $i_1 \le \ldots \le i_q$. Let $\psi$ be a 2-coloring of $G$ such that $\psi(W_{i_1})=1$, $\psi(W_{i_2})=2$, $\psi(W_{i_3})=1$, and so on, and such that the vertices in the small witness sets are all colored $1$. Observe that $\psi$ is a ${\cal W}$-compatible 2-coloring of $G$. Furthermore, any 2-coloring which 
coincides with $\psi$ on all the vertices of the big witness sets of ${\cal W}$ is ${\cal W}$-compatible. By Observation~\ref{obs:big}, the total number of vertices in big witness sets of ${\cal W}$ is at most $2k$. Hence, the probability that a random $2$-coloring $\phi$ of $G$ coincides with $\psi$ on all the vertices of the big witness sets of ${\cal W}$ is at least $(\frac{1}{2})^{2k}=\frac{1}{4^{k}}$. 
A random 2-coloring of $G$ is generated at each iteration of the main loop, and the probability that none of these $4^k$ random 2-colorings is ${\cal W}$-compatible is at most $(1-\frac{1}{4^k})^{4^k} \leq \frac{1}{e}$.
\qed
\end{proof}

Using the following theorem on $(n,t)$-universal sets~\cite{NSS95}, we can turn our randomized algorithm into a deterministic algorithm. The {\it restriction} of a function $f: X \rightarrow Y$ to a set $S \subseteq X$ is the function $f_{|S}: S \rightarrow Y$ such that $f_{|S}(s)=f(s)$ for all $s \in S$. An $(n,t)$-{\it universal set} ${\cal F}$ is a set of functions from $\{1, 2, \ldots, n\}$ to $\{1,2\}$ such that, for every $S \subseteq \{1, 2, \ldots, n\}$ with $|S|=t$, the set ${\cal F}_{|S} = \{f_{|S} \mid f \in {\cal F}\}$ is equal to the set $2^S$ of all the functions from $S$ to $\{1,2\}$.

\begin{theorem}[\cite{NSS95}]
\label{thm:universal}
There is a deterministic algorithm that constructs an $(n,t)$-universal set ${\cal F}$ of size $2^{t+O(\log ^2 t)} \log n$ in time $2^{t+O(\log ^2 t)} \log n$.
\end{theorem}

\begin{theorem}
\label{thm:path-notrandom} 
{\sc Path-Contractibility} can be solved in time $4^{k+O(\log^2k)} + n^{O(1)}$.
\end{theorem}

\begin{proof}
Given an instance $(G,k)$ of {\sc Path-Contractibility}, we construct an equivalent instance $(G',k)$ such that $G'$ has at most $5k+3$ vertices, which can be done in $n^{O(1)}$ time by Theorem~\ref{thm:path-linear-kernel}. We then construct a $(5k+3,2k)$-universal set ${\cal F}$, which can be done in time $2^{2k+O(\log^2 2k)} \log (5k+3)=4^{k+O(\log^2 k)}$ by Theorem~\ref{thm:universal}. Note that each function in ${\cal F}$ can be interpreted as a 2-coloring of $G'$. We now run the algorithm described in the proof of Theorem~\ref{thm:path-random}, but we let the main loop run over all 2-colorings $\phi \in {\cal F}$ instead of $4^k$ random 2-colorings. In each iteration, we perform the procedure described in the proof of Lemma~\ref{lem:pathalg1} on one of the 2-coloring in ${\cal F}$, which takes $k^{O(1)}$ time. The number of iterations of the main loop is the size of ${\cal F}$, which is $4^{k+O(\log^2k)}$ by Theorem~\ref{thm:universal}. Hence the total running time is $4^{k+O(\log^2k)} k^{O(1)} + n^{O(1)}= 4^{k+O(\log^2k)} + n^{O(1)}$.
\qed
\end{proof}

\section{\sc Tree-Contractibility}
\label{sec:tree}

Asano and Hirata~\cite{AH83} showed that {\sc Tree-Contractibility} is NP-complete. In this section, we first show that, unlike {\sc Path-Contractibility}, {\sc Tree Contractibility} does not have a polynomial kernel, unless coNP $\subseteq$ NP/poly. We then go on to present a $4.88^k n^{O(1)}$ time algorithm for {\sc Tree-Contractibility}.

A {\em polynomial parameter transformation} from a parameterized problem $Q_1$ to a parameterized problem $Q_2$ is a polynomial time reduction from $Q_1$ to $Q_2$ such that the parameter of the output instance is bounded by a polynomial in the parameter of the input instance. Bodlaender et al.~\cite{BTY09} proved that if $Q_1$ is NP-complete, $Q_2$ is in NP, there is a polynomial parameter transformation from $Q_1$ to $Q_2$, and $Q_2$ has a polynomial kernel, then $Q_1$ has a polynomial kernel.

\begin{theorem}
\label{no-poly-kernel}
{\sc Tree-Contractibility} does not have a kernel with size polynomial in $k$, unless coNP $\subseteq$ NP/poly.
\end{theorem}

\begin{proof}
We give a polynomial parameter transformation from {\sc Red-Blue Domination} to {\sc Tree-Contractibility}. {\sc Red-Blue Domination} takes as input a bipartite graph $G=(A,B,E)$ and an integer $t$, and the question is whether there exists a subset of at most $t$ vertices in $B$ that dominates $A$. We may assume that every vertex of $A$ has a neighbor in $B$, and that $t \le |A|$. This problem, when parameterized by $|A|$, has been shown not to have a polynomial kernel, unless coNP $\subseteq$ NP/poly~\cite{DLS09}. Since {\sc Tree-Contractibility} is in NP, the existence of the polynomial parameter transformation described below implies that {\sc Tree-Contractibility} does not have a kernel with size polynomial in $k$, unless coNP $\subseteq$ NP/poly.

Given an instance of {\sc Red-Blue Domination}, that is a bipartite graph $G=(A,B,E)$ and an integer $t$, we construct an instance $(G',k)$ of {\sc Tree-Contractibility} with $G'=(A', B', E')$ as follows. To construct $G'$, we first add a new vertex $a$ to $A$ and make it adjacent to every vertex of $B$. We define $A'=A \cup \{a\}$. We then add, for every vertex $u$ of $A$, $k+1$ new vertices to $B$ that are all made adjacent to exactly $u$ and $a$. The set $B'$ consists of the set $B$ and the $|A|(k+1)$ newly added vertices. Finally, we set $k=|A|+t$. This completes the construction. Observe that $k \le 2|A|$, which means that the construction is parameter preserving. In particular, we added $|A|(k+1)+1 \leq 2|A|^2 + |A| +1$ vertices to $G$ to obtain $G'$, and we added $|B|$ edges incident to $a$ and then two edges incident to each vertices of $B'\setminus B$. Hence the size of the graph has grown by $O(|B|+|A|^2)$. We now show that that there is a subset of at most $t$ vertices in $B$ that dominates $A$ in $G$ if and only if $G'$ is $k$-contractible to a tree.

Assume there exists a set $S\subseteq B$ of size at most $t$ such that $S$ dominates $A$ in $G$. Vertex $a$ is adjacent to all vertices of $S$, so the set $X=\{a\}\cup S\cup A$ is connected in $G'$. Note that all the vertices of $G'$ that do not belong to $X$ form an independent set in $G$. Consider the unique witness structure of $G'$ that has $X$ as its only big witness set. Contracting all the edges of a spanning tree of $G[X]$ yields a star. Since $X$ has at most $1+t+|A|=1+k$ vertices, any spanning tree of $G[X]$ has at most $k$ edges. Hence $G'$ is $k$-contractible to a tree.

For the reverse direction, assume that $G'$ is $k$-contractible to a tree $T$, and let ${\cal W}$ be a $T$-witness structure of $G'$. Vertex $a$ is involved in $k+1$ different cycles with each vertex of $A$ through the vertices of $B' \setminus B$. Hence, if $a$ and a vertex $u$ of $A$ appear in different witness sets, we need more than $k$ contractions to kill the $k+1$ cycles containing both $a$ and $u$. Consequently, there must be a witness set $W\in {\cal W}$ that contains all the vertices of $A \cup \{a\}$. Since all the vertices of $G' - W$ belong to $B'$, they form an independent set in $G'$. This means that $W$ is the only big witness set of ${\cal W}$, and $T$ is in fact a star. Since $G'$ is $k$-contractible to $T$, we know that $|W|\leq k+1$ by Observation~\ref{obs:big}. Suppose $W$ contains a vertex $x\in B' \setminus B$. By construction, $x$ is adjacent only to $a$ and exactly one vertex $a'\in A$. Let $b'$ be a neighbor of $a'$ in $B$. Then we have $N_{G'}(x)\subseteq N_{G'}(b')$, so $W'=(W\setminus \{x\})\cup \{b'\}$ is connected and $|W'| \leq |W|$. The unique witness structure of $G'$ that has $W'$ as its only big witness set shows that $G'$ can be $k$-contracted to a tree $T'$ on 
at least as many vertices as $T$. Thus we may assume that $W$ contains no vertices of $B' \setminus B$. Let $S=W\setminus A'$. The set $W$ is connected and $A'$ is an independent set, so $S$ dominates $A'$. Moreover $|S|=|W|-|A|-1\leq k -|A| =t$. We conclude that $S$ is a subset of at most $t$ vertices in $B$ that dominates $A$ in $G$.
\qed
\end{proof}

Although {\sc Tree-Contractibility} does not have a polynomial kernel, we are able to give an algorithm for the problem with running time $4.88^{k} n^{O(1)}$. As with {\sc Path-Contractibility}, we first give a randomized algorithm with running time $4.88^k n^{O(1)}$ below. We then derandomize it to obtain a deterministic algorithm with the same complexity. The following lemma implies that we may assume the input graph to be 2-connected.

\begin{lemma}
\label{lem:2conn}
A graph is $k$-contractible to a tree if and only if each of its 2-connected components can be contracted to a tree, using at most $k$ edge contractions in total.
\end{lemma}

\begin{proof}
Let $G$ be a graph, and suppose that $G$ has a vertex $v$ whose deletion results in at least two connected components with vertex sets $V_1,\ldots,V_p$. For every $i\in \{1,\ldots,p\}$, let $G_i=G[V_i\cup \{v\}]$. We prove that $G$ is $k$-contractible to a tree if and only if each of the graphs $G_i$ is contractible to a tree, using at most $k$ edge contractions in total. Repeating this argument for each of the cut vertices of $G$ yields the lemma.

Suppose each of the graphs $G_i$ can be contracted to a tree $T_i$, and that at most $k$ edge contractions are used in total; 
let $E'\subseteq E(G)$ be the corresponding set of contracted edges, with $|E'|\leq k$. For every $i\in \{1,\ldots,p\}$, let $\mathcal W_i$ be a $T_i$-witness structure of the graph $G_i$, and let $W_i$ be the witness set of $\mathcal W_i$ containing $v$. Note that $W=\bigcup_{i=1}^p W_i$ is a connected set in $G$. We define a witness structure ${\cal W}$ of $G$ as follows: ${\cal W}$ contains all the witness sets of ${\cal W}_i\setminus W_i$ for every $i\in \{1,\ldots,p\}$, as well as one witness set formed by the set $W$. It is clear that ${\cal W}$ is a $T$-witness structure of $G$ for some tree $T$, and that contracting each of the edges of $E'$ in $G$ yields $T$. Since $|E'|\leq k$, $G$ is $k$-contractible to a tree.

To prove the reverse statement, suppose $G$ is $k$-contractible to a tree $T$. Let ${\cal W}$ be a $T$-witness structure of $G$, and let $W$ be the witness set of ${\cal W}$ containing $v$. Since $v$ is a cut vertex of $G$, every witness set in ${\cal W}\setminus \{W\}$ is contained in exactly one of the sets $V_i$. For every $i\in \{1,\ldots,p\}$, we define a witness structure ${\cal W}_i$ of the graph $G_i$ as follows: ${\cal W}_i$ contains every witness set of ${\cal W}$ that is contained in $V_i$, plus one witness set $W_i=(W\cap V_i)\cup \{v\}$. Since $v$ is a cut vertex of $G$ and the set $W$ is connected, each of the sets $W\cap V_i$ is connected as well. 
We conclude that each graph $G_i$ can be contracted to a tree $T_i$ by repeatedly contracting every edge that has both endpoints in the same witness set of ${\cal W}_i$. Since contracting exactly the same edges in $G$ yields the tree $T$ and $G$ is $k$-contractible to $T$, the total number of edge contractions needed to contract each $G_i$ to $T_i$ is at most $k$.
\qed
\end{proof}

The main idea for the $4.88^k n^{O(1)}$ randomized algorithm for {\sc Tree-Contract\-ibility} is similar to the algorithm for {\sc Path-Contractibility} that was presented in Section~\ref{sec:path}. We will again use $2$-colorings of the input graph, but we adapt the notion of compatibility. Let $T$ be a tree, and let ${\cal W}$ be a $T$-witness structure of a graph $G$. We say that a 2-coloring $\phi$ of $G$ is {\em compatible} with ${\cal W}$ (or ${\cal W}$-{\em compatible}) if the following two conditions are both satisfied: {\bf (1)} every witness set of ${\cal W}$ is monochromatic with respect to $\phi$, and {\bf (2)} if $W(u)$ and $W(v)$ are big witness sets and $uv\in E(T)$, then $\phi(W(u)) \neq \phi(W(v))$.

Let $\phi$ be a given 2-coloring of a 2-connected graph $G$. In Lemma~\ref{lem:treealg1}, we will show that if $\phi$ is ${\cal W}$-compatible, then we can use the monochromatic components of $G$ with respect to $\phi$ to compute a $T'$-witness structure of $G$, such that $T'$ is a tree with at least as many vertices as $T$. Informally, we do this by finding a ``star-like'' partition of each monochromatic component $M$ of $G$, where one set of the partition is a connected vertex cover of $G[M]$, and all the other sets have size 1. A {\em connected vertex cover} of a graph $G$ is a subset $V'\subseteq V(G)$ such that $G[V']$ is connected and every edge of $G$ has at least one endpoint in $V'$. While it is NP-hard to find a connected vertex cover of minimum size in a given graph~\cite{GJ79}, there is a fast fixed parameter tractable algorithm for the problem.

\begin{proposition}[\cite{BF10}]
\label{prp:cvc}
Given a graph $G$, it can be decided in $2.44^t n^{O(1)}$ time whether $G$ has a connected vertex cover of size at most $t$. If such a connected vertex cover exists, then it can be computed within the same time.
\end{proposition}

\begin{lemma}
\label{lem:treealg1}
Let $\phi$ be a $2$-coloring of a 2-connected graph $G$. If $\phi$ is compatible with a $T$-witness structure of $G$ whose largest witness set has size $d$, where $T$ is a tree, then a $T'$-witness structure of $G$ can be computed in time $2.44^d \, n^{O(1)}$, such that $T'$ is a tree with as at least as many vertices as $T$.
\end{lemma}

\begin{proof}
Suppose $\phi$ is compatible with a $T$-witness structure ${\cal W}$ of $G$, such that $T$ is a tree, and the largest witness set of ${\cal W}$ has size $d$. We start by making a simple observation about the witness sets of ${\cal W}$. Suppose $W(v)$ is a witness set of ${\cal W}$ for some vertex $v\in V(T)$ that has degree at least 2 in $T$. Then $W(v)$ must be big, since otherwise the single vertex of $W(v)$ would be a cut vertex of $G$, contradicting the assumption that $G$ is 2-connected. This implies that if a witness set $W(v')\in {\cal W}$ is small, then $v'$ is a leaf of $T$. 

Let ${\cal X}$ be the set of monochromatic components of $G$ with respect to $\phi$. Every witness set of ${\cal W}$ is monochromatic by property (1) of a ${\cal W}$-compatible 2-coloring, and connected by definition. Hence, for every $W\in {\cal W}$, there exists an $X\in {\cal X}$ such that $W\subseteq X$. Moreover, since every $X\in {\cal X}$ is connected, there exists a vertex subset $Y\subseteq V(T)$ such that $T[Y]$ is a connected subtree of $T$ and $X=\bigcup_{y\in Y} W(y)$. Hence, ${\cal X}$ is a $T''$-witness structure of $G$ for a tree $T''$ that has at most as many vertices as $T$. We now show how to partition the big witness sets of ${\cal X}$ in such a way, that we obtain a $T'$-witness structure of $G$ for some tree $T'$ on at least as many vertices as $T$.

Suppose there exists a set $X\in {\cal X}$ that contains more than one witness set of ${\cal W}$, say $W(v_1),\ldots,W(v_p)$ for some $p\geq 2$. As a result of the observation we made earlier and properties (1) and (2) of a ${\cal W}$-compatible 2-coloring, we know that at most one of those sets can be big. If all the sets $W(v_1),\ldots,W(v_p)$ are small, then all the vertices $v_1,\ldots,v_p$ are leaves in $T$. This means that $p=2$ and $T$ consists of only two vertices; a trivial case. Suppose one of the sets, say $W(p_1)$, is big. Since each of the sets $W(v_2),\ldots,W(v_p)$ is small, the vertices $v_2,\ldots,v_p$ are leaves in $T$. This means that the vertices $v_1,\ldots,v_p$ induce a star in $T$, with center $v_1$ and leaves $v_2,\ldots,v_p$. Note that $W(v_1)$ is a connected vertex cover in the graph $G[X]$; this observation will be used in the algorithm below. Also note that the sets $W(v_1),\ldots,W(v_p)$ define an $S$-witness structure ${\cal S}$ of the graph $G[X]$, where $S$ is a star with $p-1$ leaves. 

We use the above observations to decide, for each $X\in {\cal X}$, if we can partition $X$ into several witness sets. Recall that, given $\phi$, we only know ${\cal X}$, and not ${\cal W}$. We perform the following procedure on each set $X\in {\cal X}$ that contains more than one vertex. Let $\hat{X} = X \cap N_G(V \setminus X)$ be the set of vertices in $X$ that have at least one neighbor outside $X$. A {\em star-shatter} of $X$ is a partition of $X$ into sets, such that one of them is a connected vertex cover $C$ of $G[X]$ containing every vertex of $\hat{X}$, and each of the others has size 1. The {\em size} of a star-shatter is the size of $C$. A star-shatter of $X$ of minimum size can be found as follows. Recall that we assume the largest witness set of ${\cal W}$ to be of size $d$. Construct a graph $G'$ from the graph $G[X]$ by adding, for each vertex $x\in \hat X$, a new vertex $x'$ and an edge $xx'$. Find a connected vertex cover $C$ of minimum size in $G'$ by applying the algorithm of Proposition~\ref{prp:cvc} for all values of $t$ from $1$ to $d$. Since $\phi$ is ${\cal W}$-compatible and all witness sets of ${\cal W}$ have size at most $d$, such a set $C$ will always be found. Observe that a minimum size connected vertex cover of $G'$ does not contain any vertex of degree 1, which implies that $\hat X\subseteq C$. Hence $C$, together with the sets of size 1 formed by each of the vertices of $X\setminus C$, is a minimum size star-shatter of $X$. If, in ${\cal X}$, we replace $X$ by the sets of this minimum size star shatter of $X$, we obtain a $\tilde{T}$-witness structure of $G$, for some tree $\tilde{T}$ on strictly more vertices than $T''$. We point out that the size of $C$ is at most as big as the size of the only possible big witness set of ${\cal W}$ that $X$ contains. Hence, after repeating the above procedure on each of the sets of ${\cal X}$ that contain more than one vertex, we obtain a desired $T'$-witness structure of $G$, where $T'$ is a tree on at least at many vertices as $T$.

By Proposition~\ref{prp:cvc}, we can find a minimum size star-shatter in $2.44^d \, n^{O(1)}$ time for each set of ${\cal X}$. Since all the other steps can be performed in polynomial time, the overall running time is $2.44^d \, n^{O(1)}$.
\qed
\end{proof}

\begin{theorem}
\label{thm:tree}
There is a Monte Carlo algorithm with running time $4.88^k \, n^{O(1)}$ for {\sc Tree-Con\-tract\-ibility}. If the input graph $G$ is not $k$-contract\-ible to a tree, then the algorithm correctly outputs {\sc no}. If $G$ is $k$-contractible to a tree, then the algorithm correctly outputs {\sc yes} with probability at least $1-\frac{1}{e}$. Moreover, if it outputs {\sc yes}, it also outputs a $T$-witness structure, for a tree $T$, such that $G$ is $k$-contractible to $T$.
\end{theorem}

\begin{proof}
We describe such a randomized algorithm. Let $G$ be an input graph. The algorithm has an outer loop, which guesses the size $d$ of the largest witness set of a possible $T$-witness structure ${\cal W}$ of $G$ for a tree $T$. In particular, the outer loop iterates over all values of $d$ from $1$ to $k+1$. For each value of $d$, the algorithm runs the following inner loop, which is repeated $2^{2k-d-2}$ times. At each iteration of the inner loop, the algorithm generates a random 2-coloring $\phi$ of G. It then computes a minimum size star-shatter for each of the monochromatic components of $G$ with respect to $\phi$, using the $2.44^d \, n^{O(1)}$ time procedure described in the proof of Lemma~\ref{lem:treealg1} with the value $d$ determined by the outer loop. If this procedure yields a $T'$-witness structure of $G$ for a tree $T'$ on at least $n-k$ vertices at some iteration of the inner loop, then the algorithm returns this witness structure, answers {\sc yes}, and terminates. If none of the iterations of the inner loop yields such a witness structure, the outer loop picks the next value of $d$. If none of the iterations of the outer loop yields a $T'$-witness structure of $G$ for a tree $T'$ on at least $n-k$ vertices, then the algorithm returns {\sc no}. Observe that 
$\sum_{d=1}^{k+1} 2^{2k-d-2}2.44^d\leq 4.88^k$. It follows that the total running time of the algorithm is $4.88^k \, n^{O(1)}$.

Suppose $G$ is $k$-contractible to a tree $T$, and let $\mathcal W$ be a $T$-witness structure of $G$ whose largest witness set has size $d$. Note that $d\leq k+1$ by Observation~\ref{obs:big}. Let $\psi$ be a 2-coloring of $G$ such that each of the big witness sets of ${\cal W}$ is monochromatic with respect to $\psi$, such that $\psi(W(u))\neq \psi(W(v))$ whenever $uv$ is an edge in $T$, and such that the vertices in the small witness sets are all colored 1. Observe that $\psi$ is a ${\cal W}$-compatible 2-coloring of $G$, as is any other 2-coloring of $G$ that coincides with $\psi$ on all the vertices of the big witness sets of ${\cal W}$. 

As a result of Observation 1, the number of vertices contained in big witness sets is at most $d+2(k-d-1)=2k-d-2$. Hence, the probability that a random $2$-coloring $\phi$ coincides with $\psi$ on all vertices contained in big witness sets, and thus is ${\cal W}$-compatible, is at least $\frac{1}{2^{2k-d-2}}$. Recall that at every iteration of the outer loop, a random 2-coloring of $G$ is generated at each iteration of the inner loop. The probability that none of the $2^{2k-d-2}$ random 2-colorings generated by the inner loop is ${\cal W}$-compatible, when the value of $d$ in the outer loop is correct, is at most $\bigl(1-\frac{1}{2^{2k-d-2}}\bigr)^{2^{2k-d-2}} \leq \frac{1}{e}$.
\qed
\end{proof}

Using Theorem~\ref{thm:universal}, we can derandomize the randomized algorithm for {\sc Tree-Contract\-ibility} in the same way we did for {\sc Path-Contractibility}.

\begin{theorem}
\label{thm:treedeterministic}
{\sc Tree-Contractibility} can be solved in time $4.88^{k}\,  n^{O(1)}$.
\end{theorem}

\begin{proof}
Let $(G,k)$ be an instance of {\sc Tree-Contractibility}. For each value of $d$, we construct an $(n,2k-d-2)$-universal set $\cal F$, and let the inner loop iterate over all 2-colorings $\phi \in \cal F$ instead of randomly generated $2$-colorings. For each $d$, the size of ${\cal F}$ is $2^{2k-d-2 + \log^2 (2k-d-2)} \log n$, and ${\cal F}$ can be constructed in $2^{2k-d-2 + \log^2 (2k-d-2)} \log n$ time, by Theorem~\ref{thm:universal}. Summating over all values of $d$ from 1 to $k+1$, just as we did in the proof of Theorem~\ref{thm:tree}, shows that this deterministic algorithm runs in time $4.88^k \, n^{O(1)}$.
\qed
\end{proof}

\section{Concluding remarks}

The number of edges to contract in order to obtain a certain graph property is a natural measure of how close the input graph is to having that property, similar to the better established similarity measures of the number of edges to delete and the number of vertices to delete. The last two measures are well studied when the desired property is being acyclic or being a path, defining some of the most widely known and well studied problems within Parameterized Complexity. Inspired by this, we gave kernelization results and fast fixed parameter tractable algorithms for {\sc Path-Contractibility} and {\sc Tree-Contractibility}, parameterized by the similarity measure of the number of contracted edges. We think our results motivate the parameterized study of similar problems, an example of which is {\sc Interval-Contractibility}. Modifying a given graph to an interval graph has applications in computational biology. The similarity measures of the input graphs are often small in practice, which makes the problem well suited for parameterized studies. We conclude by posing the fixed parameter tractability of {\sc Interval-Contractibility} as an open question.

\bigskip
\noindent
{\bf Acknowledgment.} The authors are indebted to Saket Saurabh for valuable suggestions and comments.

\end{document}